\documentclass[10pt]{article}
\usepackage[latin1]{inputenc} 
\usepackage[T1]{fontenc}
\usepackage{amsmath,amssymb,amsfonts,amsthm}
\usepackage{hyperref}

\newtheorem{dfn}{Definition}
\newtheorem{prop}{Proposition}

\newtheorem{cor}{Corollary}

\let\leq\leqslant
\let\geq\geqslant

\makeatletter
\def\@fnsymbol#1{\ensuremath{\ifcase#1\or *\or (a)\or (b)\or (c)\or (d)\or \S \or
   \mathsection\or \mathparagraph\or \|\or **\or \dagger \or \ddagger \or \dagger\dagger
   \or \ddagger\ddagger \else\@ctrerr\fi}}
\makeatother

\title{\textbf{Playing with parameters: structural parameterization in graphs}\thanks{This work was supported the French Agency for Research under the DEFIS program TODO, ANR-09-EMER-010}}

\author{N. Bourgeois\thanks{Universit\'e Paris~1, SAMM, \texttt{nbourgeo@phare.normalesup.org}} \and 
K.K. Dabrowski\thanks{ESSEC Business School, \texttt{\{dabrowski,demange\}@essec.edu}} \and 
M. Demange\footnotemark[3] \and
V.Th.~Paschos\thanks{PSL Research University, Universit\'e Paris-Dauphine, LAMSADE CNRS, UMR 724 and Institut Universitaire de France, \texttt{paschos@lamsade.dauphine.fr}}
}

\begin{document}

\maketitle

\begin{abstract}

When considering a graph problem from a parameterized point of view, the
parameter chosen is often the size of an optimal solution of this problem (the
``standard'' parameter). A natural subject for investigation is what happens when we
parameterize such a problem by various other parameters, some of which may be
the values of optimal solutions to different problems. Such research is known
as {\em parameterized ecology}.
In this paper, we
investigate seven natural vertex problems, along with their respective
parameters:~$\alpha$ (the size of a maximum independent set),~$\tau$ (the size
of a minimum vertex cover),~$\omega$  (the size of a maximum clique),~$\chi$
(the chromatic number),~$\gamma$ (the size of a minimum dominating set),~$i$
(the size of a minimum independent dominating set) and~$\nu$ (the size of a
minimum feedback vertex set). We study the parameterized complexity of each of
these problems with respect to the standard parameter of the others. 

\end{abstract}

\section{Introduction: structural parameterization}

Parameterized complexity has been widely studied over the past few years. The
main motivation for this area is to study the tractability of a problem with
respect to the size of some of its parameters besides the size of the instance.
In particular, some NP-hard problems may become more tractable for instances
with small parameter value (see the books~\cite{dowfel,FG06,niedermeier06} for
more details about parameterized complexity).  From a practical point of view
this may be interesting when the considered parameter has a strong dependence
on the underlying model, in which case instances with low parameter value may be
relevant. Much of the time in the literature, the main parameter used
for an optimization problem was the optimal value itself.  This causes two
limitations: first, it sometimes becomes difficult to compare the parameterized
complexity of two different problems, each of them dealing with a specific
parameter. Second, it may happen that instances with small parameter value but
large optimal value are relevant and in this case an approach with this
specific parameter may allow us to solve such instances efficiently.

Another parameter that has seen much interest is that of treewidth. However,
it is useful to consider problems with respect to any kind of parameter and such study is normally referred to by the term {\em structural parameterization}.
Considering the parameterized
complexity of a problem with respect to several parameters gives more
information and deeper insight into the real tractability of the problem. It also
provides a more stable framework for comparing the tractability of different
problems. 
If we consider standard parameters for a series of problems, it could be
relevant to study parameterized complexity for each of the problems considered
with respect to the standard parameters of the others. The first such
systematic study of ``parameterized ecology'' appeared in~\cite{FLMMRS09}. They
describe a number of methods to aid in such classifications. Follow-up work to
this paper can be found in the survey~\cite{Fellows2013541} and further
advances in this direction appear in~\cite{KN12}.  Various papers have studied
problems parameterized by non-standard parameters (see e.g.~\cite{JB13,BJK13,MLPPS12}). Many more such results can be found in~\cite{Jansenthesis}.

The problems focussed on in~\cite{FLMMRS09,Fellows2013541} were vertex
cover, dominating set, treewidth, bandwidth, genus and maximum leaf number\footnote{For a
connected graph, the maximum leaf number is the maximum number of leaves of a
spanning subtree of the graph.} The first two of these problems are in some
sense natural, in that they are problems that often need to be solved in
real-world applications. While the other four problems are interesting from a
theoretical point of view as parameters, they are not, in and of themselves, of
much practical use as problems.

We study seven very well-known combinatorial optimization problems on graphs
that are very relevant for many practical applications. Each of these problems
has a standard parameter associated with it.  Our aim here is to study the
natural question of the tractability of each problem under each considered
parameter (see Table~\ref{tab:results}). However, this  could also be done for
any other combinatorial problem and any parameter. 

When handling optimization problems, three different versions of a problem can
be considered: 
\begin{enumerate}
\item computing, for any instance, an optimal solution (the {\em constructive} case);
\item calculating the optimal value (the {\em non-constructive} case); 
\item solving the decision version of a problem.
\end{enumerate}
We
include this distinction in our general framework (see Section~\ref{sec:cnc}),
defining a general notion of FPT reducibility and proving the equivalence
between these different versions for a wide range of problems including our
seven basic graph problems. Then, in Section~\ref{sec:proof}, we prove the
results of Table~\ref{tab:results}, considering first tractability and then
intractability results. Some of the proofs in this section rely on results
given in Section~\ref{sec:cnc}. 

\subsection{Notation}

In our investigation of parameterized ecology, we have selected seven basic
vertex parameters of graphs on which to focus our study:
\begin{enumerate}
\item[$\alpha$:] the size of a maximum independent set;
\item[$\tau$:] the size of a minimum vertex cover;
\item[$\omega$:] the size of a maximum clique;
\item[$\chi$:] the chromatic number;
\item[$\gamma$:] the size of a minimum dominating set;
\item[$i$:] the size of a minimum independent dominating set;
\item[$\nu$:] the size of a minimum feedback vertex set.
\end{enumerate}
We also write $n$ for the number of vertices and $\Delta$ for the maximum
degree. All these values are integer graph parameters and we will use $p$ to
refer to any general graph parameter, i.e. $p$ is any computable function that
takes a graph as input and outputs an integer value. The parameters we consider
all take only non-negative integral values. For a graph $G$, $p(G)$
denotes the value of the related parameter for $G$. 

For any such parameter~$p$, there is an associated combinatorial
problem~$\Pi^p$, which is that of computing the parameter for a given graph.
Here, the considered parameter is in fact the value of the optimization
problem~$\Pi^p$ and in Section~\ref{sec:cnc} we will distinguish between
computing the parameter itself and computing a corresponding optimal solution.
A graph problem asking for a value, solution or answer to a decision question
is said to be {\em fixed-parameter tractable} (FPT) with respect to a
parameter~$p$ (or simply FPT($p$)) if there is an algorithm that solves the
problem in~$O(g(p(G))P(n(G))$ time on input graph~$G$, where~$g$ is a
computable function and~$P$ is a polynomial function. Without loss of
generality  we assume that $g$ is non-decreasing (otherwise replace $g(p)$ by
${\max\{g(p') | p' \leq p\}}$). When no ambiguity occurs, we write
$O(g(p)P(n))=O^*(g(p))$. We say that such an algorithm is an FPT algorithm and
runs in FPT time. We will sometimes write~$(\Pi,p)$ to refer to the
problem~$\Pi$ parameterized by~$p$.

For a graph~$G$ with a vertex~$x$, the set~$N(x)$ denotes the {neighbourhood}
of~$x$, i.e. the set of vertices adjacent to~$x$. We set~$N[x] = N(x) \cup
\{x\}$, the {closed neighbourhood} of~$x$. If~$D$ is a set of vertices,~$N[D]$
denotes the union of the closed neighbourhoods of vertices in~$D$.  Finally,~$V(G)$ denotes the vertex set of~$G$ and for $V'\subset V(G)$, $G[V']$ denotes the sub-graph of~$G$ induced by~$V'$.

\subsection{Our results}


To date, much research has focused on the complexity of a problem when
parameterized by the solution size, so we can already fill the diagonal line of
Table~\ref{tab:results}.  Some problems have been shown to be fixed
parameter-tractable:~$\Pi^\tau$ and~$\Pi^\alpha$ are FPT($\tau$)~\cite{taufpt},
and~$\Pi^\nu$ is FPT($\nu$)~\cite{nufpt}.

On the other hand,~$\Pi^\omega$ (resp.~$\Pi^\alpha$) is a classic example of a
problem which is W[1]($\omega$)-complete (resp.
W[1]($\alpha$)-complete)~\cite{omegaw1}, while~$\Pi^\gamma$ and~$\Pi^i$ are
W[2]($\gamma$)-complete and W[2]($i$)-complete, respectively~\cite{gammaw2}.
Problem~$\Pi^\chi$ is~$\notin$XP($\chi$) since chromatic number remains NP-hard
when the optimum is 3~\cite{Lovasz73,Stockmeyer73}. Here, we take~$\notin$XP to
mean that the problem is not in the class~XP (assuming $\mathrm{P} \neq
\mathrm{NP}$). Note that any vertex cover consists of vertices whose removal
leaves an independent set (and vice versa). Therefore, any minimum vertex cover
consists of precisely those vertices that are not in some maximum independent
set. Therefore, any results for the complexity of~$\Pi^\alpha$ will also apply
to the complexity of~$\Pi^\tau$ and vice versa (as suggested in
Table~\ref{tab:results}).


The overall results are summarized in Table~\ref{tab:results} where, for a
complexity class~C, C-c (resp. C-h) means C-complete (resp. C-hard).

\begin{table}[h]
\begin{center}
\begin{tabular}{c|cccccc}
         & $\Pi^\omega$ & $\Pi^\chi$ & $\Pi^\gamma$ & $\Pi^i$     & $\Pi^\nu$  & $\Pi^\alpha$/$\Pi^\tau$ \\
\hline
$\omega$ & W[1]-c       & $\notin$XP & $\notin$XP   & $\notin$XP  & $\notin$XP & $\notin$XP \\
$\chi$   & W[1]-h       & $\notin$XP & $\notin$XP   & $\notin$XP  & $\notin$XP & $\notin$XP \\
$\gamma$ & $\notin$XP   & $\notin$XP & W[2]-c       & $\notin$XP  & $\notin$XP & $\notin$XP \\
$i$      & $\notin$XP   & $\notin$XP & W[2]-h       & W[2]-c      & $\notin$XP & $\notin$XP \\
$\nu$    & FPT          & FPT        & FPT          & FPT         & FPT        & FPT        \\
$\tau$   & FPT          & FPT        & FPT          & FPT         & FPT        & FPT        \\
$\alpha$ & $\notin$XP   & $\notin$XP & W[2]-h       & W[2]-h      & W[1]-h     & W[1]-c     \\
\end{tabular}
\end{center}
\caption{A summary of our results. The columns represent graph problems. The rows represent parameters.}\label{tab:results}
\end{table}

\section{Constructive vs.\ non-constructive computation}\label{sec:cnc}

Before going into the proofs of the main results reported in
Table~\ref{tab:results}, we first revisit, in the context of the present
framework, the question of equivalence between computing an optimal solution of
an optimization problem, computing its optimal value and the related decision
problem. In particular, we will make use of this equivalence in the proof of Claim~\ref{prop:gammaialpha} in Proposition~\ref{prop:alphaalphaprop}.
\begin{dfn}\label{eq:prog-math}
Any instance of an optimization problem $\Pi$ can be expressed as a
mathematical program of the form below with objective function $f$ and
constraint set $\mathcal{C}$:
\begin{equation}
\left\{\begin{array}{l}
\max {\rm \ or\ }\min  f(x)\\
x\in \mathcal{C}
\end{array}\right.
\end{equation}
\end{dfn}
When dealing with such optimization problems, several frameworks can be
considered, leading to three different versions of the problem. The {\em
constructive} version $\Pi_c$ asks us to compute an {\em optimal solution} for
the input instance, while the {\em non-constructive} (or {\em value}) version
$\Pi_v$ only asks us to compute an optimal value.  Finally, the decision
version $\Pi_d$ asks us to decide, for any value $k$, whether there is some
feasible solution $x\in\mathcal{C}$ satisfying $f(x)\geq k$ (if $\Pi$ is a
maximization problem) or $f(x)\leq k$ (if $\Pi$ is a minimization problem).
When needed, we will denote any problem with its version $\Pi_t$, with
$t\in\{c,v,d\}$; when this subscript is not specified, we will consider the
constructive version $t=c$.

This distinction, along with the relative complexity of the different versions
has been considered several times in the literature, in particular for the
classical complexity framework~\cite{paz-moran}  and for the framework of
polynomial approximation~\cite{demange-lorenzo}. This same distinction can be
considered in the frame of parameterized complexity.  Most often, negative
results are stated for the value version, while positive results are stated for
the constructive version. Note that the value version is not more difficult
than the constructive one as long as the objective function can be computed in
reasonable time. More precisely if $f$ can be computed in polynomial
(resp. FPT) time then any polynomial (resp. FPT) algorithm for the constructive
version can immediately be turned into a polynomial (resp.  FPT) algorithm for
the value version. The same holds between the value version and the decision
version, the former being at least as difficult as the latter if $f$ can be
efficiently computed.

In this paper we only consider problems for which $f$ can be computed in
polynomial time and consequently the constructive version is at least as hard
as the non-constructive one, which itself is also as hard as the decision
version. An interesting question is whether or not these versions are
equivalent in complexity.  To study the relative complexity of problems and, in
particular, of the different versions of a problem, the notion of {\em
reduction} is useful. Many kinds of reductions (mainly polynomial ones) have
been introduced in the literature, allowing us to compare tractability of
different problems and even between the different versions of a given problem.
The notion of {\em FPT reduction}, able to transfer FPT algorithms from one
problem to another one, has also been introduced (see
e.g.~\cite{dowfel,FG06,niedermeier06}). Here, we enhance the usual definition of Turing fpt reductions it in order
to
integrate the possibility of considering problems that may not necessarily be decision problems, and are specifiedand are specified  with respect to any kind of parameter.
\begin{dfn}\label{def:reduction}
Let $(\Pi^1_{t_1},p_1), (\Pi^2_{t_2},p_2)$ be two optimization problems
parameterized by $p_1$ and $p_2$ respectively, with $t_1,t_2\in\{c,v,d\}$. An
FPT-reduction from~$(\Pi^1_{t_1},p_1)$ to~$(\Pi^2_{t_2},p_2)$ is an algorithm
solving~$\mathcal{A}_1$ in FPT time with respect to parameter~$p_1$ using an
Oracle~$\mathcal{O}_2$ for~$\Pi^2_{t_2}$ such that  for any instance~$I_1$
of~$\Pi^1_{t_1}$, any call on~$\mathcal{O}_2$ is made on instances~$I'_2$
for~$\Pi^2_{t_2}$ whose size is polynomially bounded with respect to~$I_1$ and
which satisfy $p_2(I'_2)\leq h(p_1(I_1))$ for some function~$h$. We then say
that~$(\Pi^1_{t_1},p_1)$ FPT-reduces to~$(\Pi^2_{t_2},p_2)$ and denote this by
$(\Pi^1_{t_1},p_1)\leq_{FPT}(\Pi^2_{t_2},p_2)$. If the reduction is polynomial,
we denote it by $\Pi^1_{t_1}\leq_{P}\Pi^2_{t_2}$ (no parameter needs to be
specified).
\end{dfn}
Note that, since $\mathcal {A}_1$ is FPT, the number of calls to
Oracle~$\mathcal{O}_2$ is bounded by an FPT function with respect to
parameter~$p_1$ and consequently if~$\mathcal{O}_2$ is an FPT algorithm with
respect to $p_2$, then the conditions on $p_2(I'_2)$ and on the size of $I'_2$
ensures that  the reduction leads to an FPT algorithm for $\Pi^1_{t_1}$ with
respect to parameter $p_1$. In other words, such a reduction is able to
transform an FPT algorithm for $\Pi^2_{t_2}$ with respect to parameter $p_2$
into an FPT algorithm for $\Pi^1_{t_1}$ with respect to parameter $p_1$.

Note that the decision version and the non-constructive version are equivalent
for a very large class of problems, as stated by the following proposition:
\begin{prop}\label{prop:prop1}
Let $\Pi$ be an optimization problem, an instance of which is defined as in
Definition~\ref{eq:prog-math}, with parameter $p$ satisfying:
\begin{enumerate}
\item $f$ has integral values;
\item the output value associated to some feasible input can be found in polynomial time;
\item there is a polynomial function $P$ and a function $\ell$ such that $\forall
x,y\in \mathcal{C}, |f(x)-f(y)|\leq 2^{\ell(p)P(n)}$.
\end{enumerate}
Then $(\Pi_v,p)\leq_{FPT}(\Pi_d,p)$.
\end{prop}
In particular, this holds for integral non-negative objective functions
boun\-ded by $2^{\ell(p)P(n)}$.
\begin{proof}
The reduction is easily done by binary search. Without loss of generality we
assume that~$\Pi$ is a maximization problem (the minimization case is similar).
We start by finding a feasible output value. Next, we try to find a $K\leq
2^{\ell(p)P(n)+1}$ such that the optimal value lies in $[\lambda, \lambda +K]$.
Note that~$\ell(p)$ may not be explicitly known. We ask the Oracle~$\mathcal{O}$
for the problem~$\Pi_d$, whether $\exists x\in\mathcal{C}, f(x)\geq
\lambda+2^k$ for successive values of $k\geq 1$, or not.  Let~$K$ be the first
value for which the answer is NO. We know that $K\leq 2^{\ell(p)P(n)+1}$ by
hypothesis. We then find the optimal value by binary search in the interval
$[\lambda, \lambda +K]$, using Oracle $\mathcal{O}$.

This process is FPT with respect to parameter $p$ since the number of calls to
Oracle~$\mathcal{O}$ is at most $2\ell(p)P(n)+4$, each time for the same
$\Pi$-instance.
\end{proof}
For any optimization problem~$\Pi$ and parameter~$p$, we let $\Pi|_{p-{\rm
bounded}}$ denote the sub-problem of~$\Pi$ restricted to instances where $p\leq
K$ for a fixed bound~$K$. In~\cite{paz-moran} a general process was proposed to
reduce the constructive version~$\Pi_c$ of an optimization problem~$\Pi$ to its
non-constructive version~$\Pi_v$. The main idea is to transform~$\Pi$
into~$\Pi'$ by transforming the objective function~$f$ into a one-to-one
function~$f'$ such that $\forall x,y\in \mathcal{C}, f(x)\leq f(y)\Rightarrow
f'(x)\leq f'(y)$.

Moreover, one needs to suppose that the inverse function $f'^{-1}$ can be
computed in polynomial time. In particular, this holds when $f$ has integral
values, there is a polynomial-time computable bound $B$ such that
$|\mathcal{C}|\leq B$ and there is a total order on $\mathcal{C}$ such that the
related rank-function $r$ as well as its inverse function $r^{-1}$ are both
polynomially computable. Typically, for $\mathcal{C}\subset \{0,1\}^{P(n)}$,
where $P$ is a polynomial function, the lexicographic order can be computed
and inverted in polynomial time.

We can then take $f'(x)=(B+1)f(x)+r(x)$.  Note that if such a function~$f'$
does exist and if $\Pi_d\in NP$, then $\Pi'_d\in NP$.
\begin{prop}\label{pro:eq-paz-mor} Suppose $\Pi_d\in NP$, $(\Pi|_{p-{\rm
bounded}})_d$ is NP-complete for a parameter~$p$ and there is a polynomial
function~$P$ such that any instance of~$\Pi$ satisfies $\mathcal{C}\subseteq
\{0,1\}^{P(n)}$. Then $(\Pi_c,p)\leq_{FPT} (\Pi_v,p)$.
\end{prop}
\begin{proof}
We have $\Pi_c\leq_P\Pi'_c$ since both problems have exactly the same feasible
solutions and any optimal solution to $\Pi'_c$ is also optimal for $\Pi_c$.
$\Pi'_c\leq_P\Pi'_v$ by definition of $\Pi'$ (see above).

Taking the bound $B=2^{P(n)}$ we get $\Pi'_v\leq_P\Pi'_d$ and since $\Pi'_d$ is
in NP and $(\Pi|_{p-{\rm bounded}})_d$ is NP-complete, $\Pi'_d\leq_P
(\Pi|_{p-{\rm bounded}})_d$.  Given an instance of~$(\Pi_c,p)$, the reduction
(see Definition~\ref{def:reduction}) simply computes an equivalent instance of
$((\Pi|_{p-{\rm bounded}})_d,p)$, where $p\leq K$, for a constant~$K$,~$h$ is
the constant function equal to~$K$ and~${\cal O}$ is an oracle for
$(\Pi|_{p-{\rm bounded}})_d$. Note that any FPT algorithm for~$(\Pi,p)$ leads
to a polynomial-time algorithm for $(\Pi|_{p-{\rm bounded}})_d$, and
consequently such a reduction transforms an FPT algorithm into a polynomial
time one. 
\end{proof}
\begin{prop}\label{pro:cons-paz-mor}
We have $(\Pi_c,p)\leq_{FPT} (\Pi_v,p)$ for:
\begin{enumerate}
\item\label{it:paromegachi} $\Pi^\alpha, \Pi^\chi, \Pi^\gamma, \Pi^i, \Pi^\nu$
and $p\in\{\omega, \chi\}$;
\item\label{it:omegachi-alpha} $\Pi^\omega, \Pi^\chi$ and $p=\alpha$;
\item\label{it:i-gamma} $\Pi^i$ and $p=\gamma$;
\item\label{it:alphaomegachinu-gammai} $\Pi^\alpha, \Pi^\tau, \Pi^\omega,
\Pi^\chi, \Pi^\nu$ and $p\in\{\gamma,i\}$.
\end{enumerate}
\end{prop}
\begin{proof} 
All these results are the consequence of Proposition~\ref{pro:eq-paz-mor}.  For
all these problems  $\mathcal{C}\subseteq \{0,1\}^{P(n)}$ for some polynomial
function $P$ and the decision version is known to be in NP. So, we need to show
that in each case the problem $(\Pi|_{p-{\rm bounded}})_d$ is NP-complete.

\ref{it:paromegachi}.  
This follows from the inequality $\Delta(G)+1 \geq \chi(G) \geq \omega(G)$. We
simply recall that all these problems remain NP-complete on graphs of degree
3~\cite{gj}, and that 3-colouring is NP-complete~\cite{gj}.

\ref{it:omegachi-alpha}.  
Using the previous results in $\overline{G}$ we get the NP-completeness  of
$\Pi^\omega_d$ in graphs whose maximum independent set is of size 3.  A
colouring in $G$ induces a clique partition in $\overline{G}$. The decision
version of \textsc{minimum clique partition} is NP-complete in graphs of
maximum degree 3~\cite{cliquepart}, and thus in graphs of maximum clique 4.
Thus, $\Pi^\chi_d$ is NP-complete in graphs whose maximum independent set is of
size 4. 

\ref{it:i-gamma}.  
Let $G=(V,E)$ be an arbitrary non-empty graph.  Let~$G_0$ and~$G_1$ be disjoint
copies of~$G$. We form the graph~$G'$ by taking the disjoint union of~$G_0$
and~$G_1$, adding a vertex~$u$ which dominates the vertices of~$G_0$, adding a
vertex~$v$ which dominates~$G_1$ and joining~$u$ to~$v$ with an edge. Then
$\gamma(G') = 2$ (since~$\{u,v\}$ is a dominating set and there is no
dominating vertex).  Consider an independent dominating set~$I$ in~$G'$.  It
cannot contain both~$u$ and~$v$. If it contains neither, then $|I| \geq 2i(G)
\geq i(G)+1$. If it contains exactly one of~$u$ and~$v$ (without loss of
generality~$u$), then it must contain an independent dominating set for~$G_1$,
so $|I| \geq i(G)+1$. But any independent dominating set of~$G_1$, together
with~$u$ is an independent dominating set of~$G'$. Thus $i(G') = i(G) +1$, so
$\Pi^i_d$ is NP-complete, even for graphs with $\gamma = 2$.

\ref{it:alphaomegachinu-gammai}. Suppose $G=(V,E)$ is a non-empty graph. Let
$\widetilde{G}$ be the graph obtained from $G$ by adding a new vertex $v$
adjacent to all of $V$. Then we have the following:
\begin{itemize}
\item $\alpha(\widetilde{G}) = \alpha(G)$;
\item $\omega(\widetilde{G})=\omega(G)+1$;
\item $\chi(\widetilde{G})=\chi(G)+1$;
\item $\tau(\widetilde{G}) = \tau(G)+1$;
\item $\nu(\widetilde{G})=\nu(G)+1$.
\end{itemize}
The first 
three items are obvious. The 
fourth follows from the fact $\alpha + \tau = n$. For the last one, it is
clear that $\nu(\widetilde{G}) \leq \nu(G)+1$, since adding the dominating
vertex to any feedback vertex set of $G$ yields a feedback vertex set of
$\widetilde{G}$. Consider a minimal feedback vertex set $F$ of $\widetilde{G}$.
Let $v$ be the dominating vertex of $\widetilde{G}$.  We want to show that $|F|
\geq \nu(G)+1$.  If $v \in F$, then $F \setminus \{v\}$ is a feedback vertex
set of $G$, and we are done. Suppose $v \not \in F$. Then $G \setminus F$ must
be a stable set (otherwise two adjacent vertices in $G$, together with $v$
would form a $C_3$ in $\widetilde{G}$). Let $w$ be an arbitrary vertex in $F$.
Then $F \setminus \{w\}$ must be a feedback vertex set in $G$, since $N_{G
\setminus (F \setminus \{w\})}(w)$ is an independent set.  Thus $|F| \geq
\nu(G)+1$ for all minimal feedback vertex sets of $\widetilde{G}$, i.e.
$\nu(\widetilde{G}) \geq \nu(G)+1$.

Note also that $\gamma(\widetilde{G})=i(\widetilde{G})=1$ and consequently the
decision version of~$\Pi^\alpha$, $\Pi^\tau$, $\Pi^\omega$, $\Pi^\chi$
and~$\Pi^\nu$ all remain NP-complete in graphs with $\gamma=i=1$. This
completes the proof.
\end{proof}
Another case where equivalence between the constructive and
non-con\-s\-t\-ru\-c\-ti\-ve cases can be easily stated is the hereditary case.
A property $h: 2^V\Longrightarrow \{0,1\}$ for a finite set $V$ is {\em
hereditary} if $h(U')\geq h(U), \forall U'\subset U\subset V$. We then consider
an {\em hereditary maximization} problem, an instance of which can be written  
\begin{equation}\label{eq:prog-her}
\left\{\begin{array}{l}
\max   f(U)\\
h(U)=1, U \subseteq 2^V
\end{array}\right.
\end{equation}
The size of this
instance is~$|V|$ and any subset $V'\subset V$ also defines an instance
of~$\Pi$.
\begin{prop}
Let $\Pi$ be an hereditary maximization problem and consider an {\em
non-decreasing} parameter $p$, meaning that  $U'\subset U\Rightarrow p(U')\leq
p(U)$. Then $(\Pi_c,p)\leq_{FPT} (\Pi_v,p)$.
\end{prop}
\begin{proof}
Considering an instance~$V$ of~$\Pi_c$ and let $S=\emptyset$. The reduction
consists of taking any element $v\in V$ and testing the value of the instance
$V\setminus \{v\}$. If this value is smaller than the value of the instance~$V$, then
every optimal solution must include~$v$, in which case we add it to a set~$S$.
If after removing~$v$, the optimal value is unchanged, there must be an optimal
solution which does not contain~$v$, in which case we delete~$v$ from the
graph.  We continue this process, deleting vertices from the graph where
possible until all remaining vertices belong to~$S$. Then~$S$ is an optimal
solution.  In all, we will need~$O(|V|)$ requests to an Oracle~$\mathcal{O}$
for~$\Pi_v$ on sub-instances of~$V$. The hypothesis on the parameter makes the
whole process FPT with respect to~$p$ if the oracle is also FPT. This concludes
the proof.
\end{proof}
\begin{cor}
If $\Pi\in\{\Pi^\alpha, \Pi^\omega, \Pi^\nu, \Pi^\tau\}$ and $p\in\{\alpha,
\omega, \chi, \nu, \tau\}$ then $(\Pi_c,p)\leq_{FPT} (\Pi_v,p)$.
\end{cor}
%
Note also that for any constructive problem that can be solved in FPT time with
respect to a given parameter $p$, the FPT equivalence between non-constructive
and constructive versions~$(\Pi_c,p)$ and~$(\Pi_v,p)$ is trivial and
consequently, this holds for all problems considered in this paper under the
parameters~$\tau$ and~$\nu$.
%
\begin{prop}\label{prop:direct_red1}
$(\Pi_c,p)\leq_{FPT} (\Pi_v,p)$ for $\Pi = \Pi^\gamma$, $p \in \{\alpha,
\gamma, i\}$ and $\Pi = \Pi^i$, $p \in \{\alpha, i\}$.
\end{prop}
\begin{proof}
Here we need explicit reductions. We first consider the $\Pi^\gamma$ case.
Consider a graph $G=(V,E)$ and a set of vertices $V'\subset V$. We let $G_{V'}$
be the graph obtained from $G$ by adding a stable set $V''$, such that
$|V''|=|V'|$ and the edges between $V'$ and $V''$ form a perfect matching.

We then have $\gamma(G_{V'})=\gamma(G)$ if and only if there is a minimum
dominating set in~$G$ containing every vertex in  $V'$. Indeed, consider a
minimum dominating set~$DS'$ of~$G_{V'}$. If $v' \in V', v'' \in V''$ are
matched vertices then $DS'\cap\{v',v''\}\neq \emptyset$. Moreover, if~$DS'$
contains any vertices of $V''$, we can replace them by the corresponding
vertices of $V'$, to obtain a new dominating set of~$G_{V'}$ that contains
every vertex in $V'$ (and no vertex of $V''$).  Thus $\gamma(G_{V'})\geq
\gamma(G)$ and if $\gamma(G_{V'})=\gamma(G)$, we have built a minimum
dominating set in~$G$ containing every vertex in $V'$.  Conversely, suppose
that there is a minimum dominating set~$DS$ of~$G$ containing every vertex of
$V'$. It is a dominating set in~$G_{V'}$ and $\gamma(G_{V'})=\gamma(G)$, as
required. 

Now, give an oracle for~$\Pi^\gamma_v$, for any set $V' \subseteq V$, we can
decide whether or not there is a minimum dominating set containing every vertex
in~$V'$. We use this to find a vertex $v_1$ which is in some minimum dominating
set. If $v_1$ does not dominate the graph, we then find a vertex $v_2$ such
that both $v_1$ and $v_2$ are in some minimum dominating set. We continue this
process, increasing $V'$ by one vertex at a time, until we find such a set $V'$
that dominates $G$. By construction, this will be a minimum dominating set for
$G$, and we will have called the oracle at most $n$ times to construct it.

Note that for any set $V'$ that we consider during this process, $\gamma(G_{V'})\leq
\gamma(G)+1$. Also, since $|V'|\leq\gamma(G)$, we also know that $i(G_{V'})\leq
|V'| + i(G) \leq \gamma(G)+i(G)\leq 2i(G)$ and $\alpha(G_{V'})\leq |V'| +
\alpha(G) \leq \gamma(G)+\alpha(G)\leq 2\alpha(G)$. It follows that this
reduction satisfies the conditions of Definition~\ref{def:reduction} for each
of the three parameters $\alpha, \gamma, i$.

For~$\Pi^i$ we devise the following reduction for deciding whether there is a
minimum independent dominating set in~$G$ containing~$v$: consider~$G'$
obtained from~$G$ by removing~$N[v]$. Then $i(G)=i(G')+1$ if and only if~$v$
belongs to some optimal solution. Indeed, suppose $i(G)=i(G')+1$ and take a
minimum independent dominating set~$IDS'$ of~$G'$. Then $IDS'\cup\{v\}$
is an independent dominating set of~$G$ of size~$i(G)$. Conversely, if  an
optimal solution~$IDS$ contains~$v$, then $IDS\setminus\{v\}$ is an independent
dominating set of~$G'$. Note moreover that $\alpha(G')\leq \alpha(G)$ and
that $i(G')\leq i(G)+1$ for any considered graph~$G'$. Here we simply need
to iterate the process on the remaining graph. Again, this reduction satisfies
the conditions of Definition~\ref{def:reduction} for the considered parameters.
\end{proof}
To summarize, in this section we have shown the equivalence between
constructive and non-constructive optimization for all problems and parameters
considered in Table~\ref{tab:results}.

Note that if $(\Pi_c,p)\leq_{FPT} (\Pi_v,p)$ then positive (FPT) and hardness
results equivalently hold for one or the other version. Consequently, in the
following sections, we do not need to specify which version of the problem we refer to. By
Proposition~\ref{prop:prop1}, we can consider FPT results to hold for the
constructive version and hardness results to hold for the decision version.

\section{Main results}\label{sec:proof}

\subsection{Tractability results}

We first recall the following inequalities, which will be useful in many of the
proofs in this section:
\begin{eqnarray*}
&&\alpha+\tau = n\\
&&\alpha \geq i \geq \gamma\\
&&\Delta+1 \geq \chi \geq \omega\\
&&\tau \geq \nu\\
\end{eqnarray*}
\begin{prop}
	$\Pi^\nu, \Pi^\omega, \Pi^\alpha, \Pi^\tau$ and $\Pi^\chi$ are  FPT($\tau$) and FPT($\nu$).
\end{prop}
\begin{proof}
Since $\tau \geq \nu$, we only need to prove that these problems are FPT($\nu$).

The problem~$\Pi^\nu$, it is known to be FPT($\nu$)~\cite{nufpt}. For all of
the remaining problems, we start by finding a minimum feedback vertex set~$F^*$
in~$O^*(f(\nu))$ time.

Next, we consider the maximum clique problem~$\Pi^\omega$. 
A clique contains at most two vertices from the forest $V\setminus F^*$. For
each subset  $C \subset F^*$ that induces a clique, we search for two adjacent
vertices~$v,v'$ which are also adjacent to every vertex of~$C$, and add them
to~$C$. If we cannot find such a pair, we look for a single vertex adjacent to
every vertex of~$C$, and add it to~$C$. If we cannot find such a vertex, we
simply keep~$C$.  Finally, we return the largest clique constructed in this
way. This algorithm has running time bounded above by~$O^*(f(\nu)2^{\nu})$.

Next, we consider the problem~$\Pi^\tau$.  This was proved in~\cite{JB13}, which
considered the size of kernels for this problem. Using this, the result
for~$\Pi^\alpha$ follows from the identity $\alpha = n - \tau$. This result can also be
proved directly using the following simple argument.  For each subset $S \subset F^*$
which is independent, we can discard~$N(S)\cup F^*$ and use a greedy algorithm
to compute a maximum independent set on the remaining forest in polynomial
time. This algorithm has running time bounded above by~$O^*(f(\nu)2^{\nu})$.

Finally, we show how to solve $\Pi^\chi$. 
Notice that $\chi(F^*) \leq \chi(G) \leq \chi(F^*)+2 \leq \nu +2$. We take each
value of $k=1,\ldots, \nu+2$ in turn and test if $G$ has a $k$-colouring. To do
this, we first find every $k$-colouring of $F^*$ (which can be done in
$O^*(k^{\nu})$ time). For each such colouring, we test if it extends to a
$k$-colouring of $G$. To do this, we try to find a $k$ list-colouring of the
forest $V\setminus F^*$.
For a vertex $v$, the list $L(v)$ of admissible colours is precisely those that
are not used to colour any vertex in $N(v) \cap F^*$.  This
list-colouring problem can be solved in polynomial time~\cite{JS97}. Note that
the algorithm will always find a valid $k$-colouring when $k=\nu +2$. This
whole procedure runs in FPT time when parameterized by $\nu$. This completes
the proof.
\end{proof}
Suppose~$T_1$ and~$T_2$ are vertex-disjoint trees rooted at~$v_1$ and~$v_2$,
respectively. Let $T_1 \leftarrow T_2$ be the tree rooted at~$v_1$ obtained by
taking the disjoint union of~$T_1$ and~$T_2$ and then joining~$v_1$ to~$v_2$
with an edge. Note that every rooted tree can be built from its vertex set
using just this operation. Moreover, such a representation for a tree can
easily be computed in linear time.
\begin{prop}
	$\Pi^\gamma$ and $\Pi^i$ are FPT($\nu$) and FPT($\tau$).
\end{prop}
\begin{proof}
Again, we only need to prove that the problems are FPT($\nu$).  We start with
the~$\Pi^\gamma$ case.  Consider a graph $G=(V,E)$. Since $\Pi^\nu \in$
FPT($\nu$), we can compute a feedback vertex set~$F^*$ with running
time~$O^*(f(\nu))$. Fix some subset $D \subset F^*$, that we assume to be the
restriction of the minimum dominating set to~$F^*$.

We now run a dynamic programming algorithm. Note that $G[V\setminus F^*]$ is a
forest. For a rooted tree $T$ (with root vertex $v$) which is a subtree of a
tree in $G[V\setminus F^*]$, a set $S \subseteq F^*$ and a value $d \in
\{0,1,2\}$, we define $A(T,S,d)$ to be the minimum size of a set $D' \in V(T)$
such that
$N[D\cup D']\cap F^*=S$, and:
\begin{itemize}
\item if $d=0$, $D'$ includes the vertex $v$ and $D \cup D'$ dominates every
vertex in $T$;
\item if $d=1$, $D'$ does not include the vertex $v$, but $D \cup D'$
dominates every vertex in $T$;
\item if $d=2$, $D \cup D'$ dominates every vertex in $T\setminus \{v\}$,
but does not dominate $v$.
\end{itemize}
If, for some choice of $T,S,d$ no such set $D'$ exists, we set $A(T,S,d)=\infty$.

Now, for each tree~$T$ in the forest $G \setminus F^*$, we choose an arbitrary
root vertex~$v$ and find a decomposition of~$T$ using~$\leftarrow$ operations.

Let~$T'$ be a subtree (rooted at~$v'$) of~$T$ that occurs in the decomposition.
We will show how to calculate the value of $A(T',S,d)$ for every possible
choice of~$S$ and~$d$.

First, as a base case, we suppose that $T'$ consists of only a single
vertex~$v'$. If $d=0$, this corresponds to the case where $D'=\{v'\}$. Thus we
set $A(T',S,0)=1$ if $S=N[D \cup \{v'\}]\cap F^*$ and $A(T',S,0)=\infty$
otherwise. If $d=1$, this corresponds to the case where $D'=\emptyset$ and~$v'$
has a neighbour in~$D$. Thus if~$v'$ has a neighbour in~$D$, we set
$A(T',N[D]\cap F^*,1)=0$. If~$v'$ has no neighbour in~$D$ or $S \neq N[D]\cap
F^*$, we set $A(T',S,1)=\infty$. If $d=2$, this corresponds to the case where
$D'=\emptyset$ and~$v'$ does not have a neighbour in~$D$. Thus if~$v'$ has no
neighbour in~$D$, we set $A(T',N[D]\cap F^*,2)=0$. If~$v'$ has a neighbour
in~$D$ or $S \neq N[D]\cap F^*$, we set $A(T',S,2)=\infty$.

Now suppose that $T'$ (rooted at $v'$) contains more than one vertex. Then $T'
= T_1 \leftarrow T_2$ for some $T_1,T_2$ rooted at $v_1,v_2$ respectively, say.
Note that $v_1=v'$, by definition of $\leftarrow$. We now show how to calculate
$A(T',S,d)$ given the values for $A(T_1,S',d')$ and $A(T_2,S',d')$ for all
possible choices of $S'$ and $d'$.

If $d=0$, this corresponds to the case where~$D'$ contains~$v'$. Consider the
restriction of~$D'$ to~$T_1$. We must have that~$T_1$ is dominated by $(D' \cap
V(T_1)) \cup D$ and that it contains the root vertex~$v_1$. In other words,
this restriction must correspond to the $A(T_1,S_1,0)$ case for some~$S_1$. In
this case, any valid~$D'$ for~$T'$ must dominate all of~$T_1$ and all of~$T_2$.
Since $v_1 \in D'$, we know that~$v_2$ is dominated by~$v_1$. Now consider the
restriction of~$D'$ to~$T_2$. We must have that $(D' \cap V(T_2)) \cup D$
dominates $V(T_2)\setminus \{v_2\}$. Since~$v_1$ is adjacent to~$v_2$ and is
present in~$D'$, vertex~$v_2$ may or may not be present in~$D'$ and it may or
may not be dominated by $(D' \cap V(T_2)) \cup D$. This corresponds to the
$A(T_1,S_2,d')$ case for some~$S_2$ and some $d' \in \{0,1,2\}$. We therefore
set $A(T',S,0)$ to be the minimum of $\{A(T_1,S_1,0)+A(T_2,S_2,d') | d' \in
\{0,1,2\},S_1 \cup S_2 = S\}$.

If $d=1$, this corresponds to the case where $v_1 \not \in D'$, but $v_1$ is
dominated by a member of $D'$. This dominating vertex must either be in $D \cup
(V(T_1) \cap D')$ or it must be $v_2$. The restrictions of $D'$ to $T_1$ and
$T_2$ therefore correspond to $A(T_1,S_1,1)$ and $A(T_2,S_2,d')$ for some $d'
\in \{0,1\}$ and some $S_1$ and $S_2$ or they correspond to $A(T_1,S_1,2)$ and
$A(T_2,S_2,0)$ for some $S_1$ and $S_2$.  Therefore, $A(T',S,1)$ is the minimum
of $\{A(T_1,S_1,d_1)+A(T_2,S_2,d_2) | (d_1,d_2) \in \{(1,0),(1,1),(2,0)\}, S_1
\cup S_2 = S\}$.

If $d=2$, this corresponds to the case where $T' \setminus \{v_1\}$ is
dominated by $D \cup D'$, but~$v_1$ is not dominated by $D \cup D'$. This means
that neither~$v_1$ nor~$v_2$ are present in~$D'$. The restriction of~$D'$
to~$T_1$ must be such that $D \cup (D' \cap V(T_1))$ does not dominate~$v_1$,
which corresponds to the $A(T_1,S_1,2)$ case, for some~$S_1$. However,~$v_2$
must be dominated by a vertex in $D \cup (D' \cap V(T_2))$. This corresponds to
the $A(T_2,S_2,1)$ case. Thus $A(T',S,2)$ is the minimum of
$\{A(T_1,S_1,2)+A(T_2,S_2,1) | S_1 \cup S_2 = S\}$.

Using the above recursion, we can calculate the value of $A(T',S,d)$ for every
rooted tree $T'$ in $G[V(G) \setminus F^*]$ in FPT time (with parameter $\nu$).
We label these trees $T_1, \ldots, T_k$.

Now, for $S \subseteq F^*$ and $i \leq k$, let  $B(i,S)$ be the size of the
smallest set $D' \in V(T_1) \cup \cdots \cup V(T_i)$ such that $N[D \cup D']
\cap (F^*\cup V(T_1) \cup \cdots \cup V(T_i)) = S \cup V(T_1) \cup \cdots \cup
V(T_i)$. Note that $B(0, S) = 0$ if $S=N[D] \cap F^*$ and $\infty$ otherwise.
Furthermore, for $i\geq 1, B(i,S)$ is the minimum of $\{B(i-1,S_1)+A(T_i,S_2,d)
| S_1 \cup S_2 = S, d \in \{0,1\}\}$. The minimum size of a dominating set
whose intersection with $F^*$ is $D$ is then $|D|+B(k,F^*)$. All these
calculations can be done in FPT time with parameter $\nu$. We thus branch over
all possible choices of $D$ and then, for each such choice, find the size of
the minimum dominating set whose intersection with $F^*$ is $D$. 


The argument for $\Pi^i$ is similar. Again, we start by finding a minimal
feedback vertex set $F^*$, but this time, we only consider $D \subseteq F^*$
that are independent. We define $A'(T,S,d)$ in the same way as $A(T,S,d)$,
except that now we only consider sets $D'$ such that $D \cup D'$ is
independent.

We now explain how to calculate $A'(T,S,d)$ for the tree~$T$, rooted at~$v$.
Again, we first consider the case where~$T$ contains a single vertex.  If
$d=0$, this corresponds to the case where $D' = \{v\}$. $D \cup D'$ must be
independent, so~$v$ cannot have any neighbours in~$D$. Therefore, we set
$A'(T,S,0)=1$ if $S=N[D \cup \{v\}]\cap F^*$ and~$v$ has no neighbours in~$D$.
Otherwise, we set $A'(T',S,0)=\infty$. If $d=1$, this corresponds to the case
where $D' = \emptyset$ and~$v$ has a neighbour in~$D$. We therefore set
$A'(T,S,1)=0$ if $S=N[D]\cap F^*$ and~$v$ has a neighbour in~$D$. Otherwise, we
set $A'(T,S,1)=\infty$. If $d=2$, this corresponds to the case where $D' =
\emptyset$ and~$v$ does not have a neighbour in~$D$. We therefore set
$A'(T,S,2)=0$ if $S=N[D]\cap F^*$ and~$v$ does not have a neighbour in~$D$.
Otherwise, we set $A'(T',S,2)=\infty$.

Now consider the case where~$T$ contains more than two vertices. Again, it must
be of the form $T_1 \leftarrow T_2$, where~$T_1$ and~$T_2$ are trees rooted
at~$v_1$ and~$v_2$, respectively, say. We now show how to calculate
$A'(T,S,d)$. If $d=0$, this corresponds to the case where $v_1 \in D'$. Note
that $D \cup D'$ must be independent, so $v_2 \not \in D'$. Vertex~$v_1$
dominates~$v_2$, so the restriction of~$D \cup D'$ to~$D \cup T_2$ may or may
not dominate~$v_2$. Therefore $A'(T,S,0)$ is the minimum of
$\{A'(T_1,S_1,0)+A'(T_2,S_2,d') | S_1 \cup S_2 = S, d'\ \in \{1,2\}\}$. If
$d=1$, this corresponds to the case where $v_1 \not \in D'$, but it is
dominated by either~$D$ or~$v_2$. Therefore, as for the case of~$\Pi^\gamma$,
$A'(T,S,1)$ is the minimum of $\{A'(T_1,S_1,d_1)+A'(T_2,S_2,d_2) |(d_1,d_2) \in
\{(1,0),(1,1),(2,0)\}, S_1 \cup S_2 = S\}$. Similarly, $A'(T,S,2)$ is the
minimum of $\{A'(T_1,S_1,2)+A'(T_2,S_2,1) | S_1 \cup S_2 = S\}$.

Finally, we define~$B'(i,S)$ for the~$\Pi^i$ problem as we did~$B(i,S)$ for
the~$\Pi^\gamma$ problem, except that we now demand that $D \cup D'$ is
independent. In the forest $V(G) \setminus F^*$, no vertex of any tree can be
adjacent to a vertex in different tree. Therefore, the algorithm for
calculating~$B'(i,S)$ from $A'(T,S,d)$ is identical to that
calculating~$B(i,S)$ from $A(T,S,d)$ for $\Pi^\gamma$. We complete the proof in
the same way as for~$\Pi^\gamma$.
%
\end{proof}

\subsection{Intractability results}

We now prove the negative results claimed in Table~\ref{tab:results}.
\begin{prop}\label{pro:nonXP}
The following hold:
\begin{enumerate}
\item\label{nonXP:paromegachi} $\Pi^\alpha$, $\Pi^\gamma$, $\Pi^\chi$, $\Pi^i$,
$\Pi^\nu$ are $\notin$XP($\chi$) and $\notin$XP($\omega$);
\item\label{nonXP:omegachi-alpha} $\Pi^\omega$ and $\Pi^\chi$ are
$\notin$XP($\alpha$);
\item\label{nonXP:i-gamma} $\Pi^i$ is $\notin$XP$(\gamma)$;
\item\label{nonXP:alphaomegachinu-gammai}  $\Pi^\alpha, \Pi^\tau, \Pi^\omega,
\Pi^\chi$ and $\Pi^\nu$ are $\notin$XP$(\gamma)$ and $\notin$XP$(i)$.
\end{enumerate}
\end{prop}
\begin{proof}
All these results are a consequence of the fact that these problems  remain
NP-hard if the related parameter is bounded, as shown in the proof of
Proposition~\ref{pro:cons-paz-mor}. 
\end{proof}
%
%
%
\begin{prop}\label{prop:alphaalphaprop}
The following hold:
\begin{enumerate}
\item\label{uno} $\Pi^\omega$ is W[1]($\chi$)-hard;
\item\label{prop:alphaalpha} $\Pi^\alpha, \Pi^\tau$ and $\Pi^\nu$ are
W[1]($\alpha$)-hard;
\item\label{prop:gammaialpha} $\Pi^\gamma$ and $\Pi^i$ are W[2]-hard$(\alpha)$ and W[2]-hard$(i)$.
\end{enumerate}
\end{prop}
\begin{proof}
Claim~\ref{uno} follows immediately from the W[1]-completeness of {\sc
multi-co\-lour clique}. This problem asks: given a graph~$G$, and a colouring
of~$G$ with~$k$ colours, does~$G$ contain a clique on~$k$ vertices? Given an
instance~$(G,k)$ of {\sc multi-colour clique}, we delete any edges both of
whose endpoints are the same colour. The resulting graph~$G'$ has $\chi \leq
k$. Our instance of {\sc multi-colour clique} is a yes-instance if and only
if~$G'$ has a clique on~$k$ vertices.

For Claim~\ref{prop:alphaalpha}, let $G=(V,E)$ be a graph and
$\overline{G}=(V,\overline{E})$ its complement, (i.e. $e \in E \Leftrightarrow e
\notin \overline{E}$). Since $\alpha(G) = \omega(\overline{G})$ and $\omega(G)
= \alpha(\overline{G})$, the result for~$\Pi^\alpha$ and~$\Pi^\tau$ is an
immediate consequence of $\Pi^\omega \in$W[1]($\omega$)-hard.

We now prove the $\Pi^\nu$ case. Given a graph $G=(V,E)$, we define $G'$ to be
the product of $G$ with a single edge, i.e. $G'=(V',E')$, where $V' = V_1 \cup
V_2$ and $E' = \{ ((v,1),(v,2)) | v \in V\} \cup \{((v,i),(u,i)) | i \in
\{1,2\}, uv \in E\}$.

We claim that $\alpha(G')=\alpha(G), \nu(G')=2(|V| - \alpha(G))$.  It is
straightforward to verify that $\alpha(G')=\alpha(G)$. Moreover, for any graph
$G$ of order $n$, we have that $\frac{n-\nu(H)}{2} \leq \alpha(H)$, so $\nu(G')
\geq |V(G')| -2\alpha(G') = 2(|V| - 2\alpha(G))$. On the other hand, for a
stable set $S$ of $G$, $S \times \{1,2\}$ induces a forest (in fact a matching)
and consequently $\nu(G') \leq 2(|V| -\alpha(G))$. This completes the proof,
since $\Pi^\alpha$ is W[1]($\alpha$)-hard.

Let us now prove Claim~\ref{prop:gammaialpha}. Since $i \leq \alpha$, we need only prove that the problems are
W[2]-hard$(\alpha)$. We prove the result for the problem $\Pi^i$.
Given a graph $G=(V,E)$, and $k \in \{1,\ldots,|V|\}$, we
define $G_{k} = (V_{k},E_{k})$, where $V_{k} = V_1\cup \cdots \cup V_k$, $V_i =
V \times \{i\}$ and $E_{k} = E_1\cup \cdots \cup E_k  \cup E'$, where
$(V_i,E_i)$ induce cliques and $E' = \{((u,i),(v,j)), i,j \in \{1, \ldots ,k\},
i \neq j, v \in N_G[u]$\}.

The following claims then hold:
\begin{enumerate}
\item if $D_{k}$ is an independent dominating set of $G_{k}$ and $D_{k}
\cap V_j = \emptyset$ for some $j \in \{1,\ldots,k\}$ then $\{v| \exists i,
(v,i) \in
D_{k}\}$ is an independent dominating set of~$G$;
\item $\alpha(G_{k}) \leq k$;
\item if $k \geq i(G)$ then $i(G_{k})=i(G)$.
\end{enumerate}
The first two claims are obvious. In order to prove the third claim, note that if
$\{a_1,\ldots,a_{i(G)}\}$ is an independent dominating set in $G$, then
$\{(a_j,j), j=1,\ldots,i(G)\}$ is an independent dominating set for $G_k$.
Applying the first claim completes the proof.

Thus, given an oracle $\mathcal{O}$ for $(\Pi^i,\alpha)$, we find a minimum
independent dominating set for $G_1,G_2,\ldots$ until such a set has no vertex
in some $V_i$. (We can do this since the constructive and non-constructive
versions are equivalent, due to Proposition~\ref{prop:direct_red1}.) The
process will finish for $G_k, k\leq i(G)+1$ and since for $j \leq i(G)+1$, we have
that $\alpha(G_j) \leq j \leq i(G)+1$, we find that
$(\Pi^i,i)\leq_{FPT}(\Pi^i,\alpha)$.  Thus $\Pi^i$ is indeed
W[2]-hard($\alpha$). The corresponding result for $\Pi^\gamma$ follows
similarly.
\end{proof}

\section{Conclusion}

We have studied the cross-parameterization of \textsc{min vertex cover},
\textsc{max independent set}, \textsc{max clique}, \textsc{min coloring},
\textsc{min dominating set}, \textsc{min independent dominating set} and
\textsc{min feedback vertex set}. We are aware of the fact that most of the
parameters handled in this paper cannot determined in FPT time and that our
study is limited to only seven problems and parameters. However, our goal was
rather structural than purely algorithmic. We have tried to show that
cross-parameterization provides a somewhat deeper insight into the real nature of
the parameterized (in)tractability of the problems handled and helps us to
better comprehend it.

As one can see in Table~\ref{tab:results}, all of the problems tackled are FPT
with respect to both~$\tau$ (the standard parameter of \textsc{min vertex
cover}) and~$\nu$ (the standard parameter of \textsc{min feedback vertex set}).
There are however problems that are hard when parameterized by $\tau$ or $\nu$,
such as \textsc{list colouring}~\cite{listcolour-vc-hard}.

Finally, let us note that cross-parameterization can be applied to many
categories of combinatorial optimization problems, defined on several structures
(and not only on graphs). For instance, what is the parameterized complexity of
\textsc{min set cover} with respect to the standard parameter of \textsc{max
set packing} or to that of \textsc{min hitting set}? What are the complexities
of the two latter problems with respect to the two remaining parameters? 

\medskip

\noindent
\textbf{Acknowledgement.} The authors would like to thank Eunjung Kim, Bart Jansen and Florian Sikora for informing us of previous work in the area of structural parameterization.

\bibliographystyle{plain}
\bibliography{multibiblio}

\end{document}